\UseRawInputEncoding 

\documentclass[reqno]{amsart}
\usepackage{amssymb}
\usepackage{hyperref}

\newtheorem{theorem}{Theorem}
\newtheorem{proposition}[theorem]{Proposition}
\newtheorem{lemma}[theorem]{Lemma}

\theoremstyle{remark}

\numberwithin{equation}{section}

\begin{document}

\title[Wave functions for integrable particle systems]
{Wave functions for quantum integrable
particle systems via
partial confluences of multivariate hypergeometric functions}

\author{J.F.  van Diejen}

\address{
Instituto de Matem\'atica y F\'{\i}sica, Universidad de Talca,
Casilla 747, Talca, Chile}

\email{diejen@inst-mat.utalca.cl}

\author{E. Emsiz}

\address{
Facultad de Matem\'aticas, Pontificia Universidad Cat\'olica de Chile,
Casilla 306, Correo 22, Santiago, Chile}
\email{eemsiz@mat.uc.cl}

\subjclass[2010]{Primary: 35Q40; Secondary: 33C67, 35C05, 35P10, 39A14, 81Q80, 81R12.}
\keywords{quantum integrable particle models, stationary Schr\"odinger equation, analytic solutions, difference equations in the spectral variable, Calogero-Sutherland system, Toda chain, 
P\"oschl-Teller potential, Morse potential, multivariate (confluent) hypergeometric functions, Whittaker functions, hyperoctahedral symmetry}

\thanks{This work was supported in part by the {\em Fondo Nacional de Desarrollo
Cient\'{\i}fico y Tecnol\'ogico (FONDECYT)} Grants   \#  1170179 and \# 1181046.}

\date{January 2018}

\begin{abstract}
Starting from the hyperoctahedral multivariate hypergeometric function of Heckman and Opdam (associated with the $BC_n$ root system), we arrive---via partial confluent limits in the sense of Oshima and Shimeno---at solutions of the eigenvalue equations for the Toda chain with one-sided boundary perturbations of P\"oschl-Teller type  and for
the hyperbolic Calogero-Sutherland system in a Morse potential.  With the aid of corresponding degenerations of  the (bispectral dual) difference equations for the Heckman-Opdam hyperoctahedral hypergeometric function, it is deduced that the 
eigensolutions in question are holomorphic in the spectral variable.
\end{abstract}

\maketitle

\section{Introduction}\label{sec1}
In this note we construct analytic solutions of the eigenvalue equations for two fundamental quantum integrable particle models, viz.
the quantum Toda Hamiltonian with one-sided boundary perturbations of P\"oschl-Teller type
\begin{align}\label{Lt}
L_x^{\text{t}}= \sum_{1\leq j\leq n}
  \frac{\partial^2}{\partial x_j^2} & - \sum_{1\leq j<n} a_j e^{-x_j+x_{j+1}}   -a_{n-1}e^{-x_{n-1}-x_n} \\
&   -  \frac{\frac{1}{4}g_S(g_S+2g_L-1)}{ \sinh^2 \textstyle{\frac{1}{2} } (x_n) } - \frac{g_L(g_L-1)}{ \sinh^2  (x_n)}  \nonumber
\end{align}
and the hyperbolic quantum Calogero-Sutherland Hamiltonian in a Morse potential
\begin{equation}\label{Lcs}
L_x^{\text{cs}}= \sum_{1\leq j\leq n}
 \Bigl(  \frac{\partial^2}{\partial x_j^2}-  g_S e^{-x_j}  -a_n e^{-2x_j}\Bigr)  
- \sum_{1\leq j<k\leq n}   \frac{\frac{1}{2}g_M(g_M-1)}{\sinh^{2} {\frac{1}{2} }   (x_j-x_k) } .
\end{equation}
In the corresponding stationary Schr\"odinger equations the variables $x_1,\ldots ,x_n$ represent  particle positions along the line and the (possibly complex) constants $g_S,g_M,g_L$ and   $a_1,\ldots ,a_{n}$ denote coupling parameters regulating the strengths of the interactions.

The Toda Hamiltonian $L_x^{\text{t}}$ \eqref{Lt} was first introduced at the level of classical mechanics by Inozemtsev \cite{ino:finite}, who constructed a Lax-pair representation for the corresponding Hamilton flow. The question of its complete integrability was addressed within the framework of the Yang-Baxter equation in Refs. \cite{kuz-jor-chr:new,kuz:separation} (at the classical level) and in Ref.
 \cite{tsi:dynamical} (at the quantum level). Explicit partial differential operators generating a complete algebra of commuting quantum integrals for
 $L_x^{\text{t}}$ \eqref{Lt} were found by Oshima \cite{osh:completely}, whereas in Ref.  \cite{ger-leb-obl:quantum} a conjectural formula for the eigensolutions
 was proposed in the form of an integral representation.
 
For special couplings, the classical dynamics of the Calogero-Sutherland Hamiltonian with a Morse potential was studied by M. Adler  \cite{adl:some}, who determined the complete integrability and the scattering behavior of this particle system by means of a Lax-pair representation. For general couplings, the complete integrability of $L_x^{\text{cs}}$ \eqref{Lcs} was deduced at the classical level in Ref. \cite{woj:integrability} with the aid of a Lax pair and at the quantum level in Ref. \cite{hal:multivariable} by means of a recursive construction of the pertinent commuting partial differential operators.   Eigenstates in the discrete spectrum of $L_x^{\text{cs}}$ \eqref{Lcs}, which become relevant for parameter regimes in which the Morse potential allows for particle binding, were constructed near the ground state in Ref. \cite{ino-mes:discrete} and for arbitrary excitations in Ref. \cite{hal:multivariable}. In fact, it was recognized in the latter work
that the eigenfunctions in question can be understood as a multivariate generalization of classical Bessel polynomials.

Despite all progress reported on these two prominent integrants of the Toda and Calogero-Sutherland integrable families,
the challenge of providing a  mathematically rigorous construction of the eigensolutions at arbitrary spectral values
has withstood time for over three decades now, except in the standard situation of vanishing perturbation parameters (i.e. when $g_S=g_L=a_n=0$). 
More specifically, for $g_S=g_L=0$ the eigenfunctions of 
$L_x^{\text{t}}$ \eqref{Lt} amount to well-studied multivariate Whittaker functions associated with the root system $D_n$ \cite{kos:quantization,has:whittaker,goo-wal:classical,bau-oco:exponential,ger-leb-obl:new,rie:mirror}, while for 
$g_S=a_n=0$  the eigenfunctions of $L_x^{\text{cs}}$ \eqref{Lcs} are expressed in terms of the celebrated Heckman-Opdam multivariate hypergeometric function associated with the root system $A_{n-1}$ \cite{hec-sch:harmonic,saw:eigenfunctions,opd:lecture,hal-rui:recursive}. Our principal goal here is to construct the corresponding solutions of the eigenvalue equations in the situation of nonvanishing
perturbation parameters $g_S, g_L$, $a_n$, and to clear up how the presence of the P\"oschl-Teller and Morse potentials affects the structure of
the multivariate (confluent) hypergeometric functions  expressing the wave functions. In a nutshell, while the perturbations under consideration manifestly preserve the $x_n\to -x_n$ reflection symmetry in $L_x^{\text{t}}$ \eqref{Lt}
and the permutation-symmetry in $L_x^{\text{cs}}$ \eqref{Lcs}, with respect to the spectral variables the $D_n$-symmetry of the Whittaker function and the $A_{n-1}$-symmetry of the hypergeometric function  turn out to arise from breaking a stronger hyperoctahedral  symmetry (stemming from the root system $BC_n$) enjoyed by the wave functions in the presence of the perturbative potentials.

Upon assuming  that $a_j \neq 0$ ($j\in \{ 1,\ldots ,n\}$), we will  exploit  the translational freedom $x_k\to x_k+c_j$
for  $k=1,\ldots ,j$   (with $c_j\in\mathbb{C}$), so as to normalize these auxiliary coupling constants from now on at the fixed values
\begin{equation}\label{norm-coupling}
a_j\equiv 2\ \ (j=1,\ldots,n-1) \quad \text{and}\quad a_n\equiv {\textstyle \frac{1}{4}} 
\end{equation}
(unless explicitly stated otherwise).
It is well-known (cf. e.g. Refs.
\cite{ino-mes:discrete,ino:finite,die:difference,osh:completely} and references therein) that both  $L_x^{\text{t}}$ \eqref{Lt} and $L_x^{\text{cs}}$ \eqref{Lcs}
arise as degenerations  of the following hyperoctahedral  Calogero-Sutherland Hamiltonian associated with the root system $BC_n$
\cite{ols-per:quantum,hec-sch:harmonic}:
\begin{align}\label{Lbc}
L_x^{\text{bc}}= \sum_{1\leq j\leq n}&
 \Bigl(  \frac{\partial^2}{\partial x_j^2}-  \frac{\frac{1}{4}g_S(g_S+2g_L-1)}{ \sinh^2 \textstyle{\frac{1}{2} } (x_j) } - \frac{g_L(g_L-1)}{ \sinh^2  (x_j)}\Bigr)  \\
&- \sum_{1\leq j<k\leq n}  \left( \frac{\frac{1}{2}g_M(g_M-1)}{ \sinh^{2} {\frac{1}{2} }   (x_j+x_k) }+ \frac{\frac{1}{2}g_M(g_M-1)}{\sinh^{2} {\frac{1}{2} }   (x_j-x_k) }\right) .
\nonumber 
\end{align}
Indeed, after a translation of
the position vector
 $x=(x_1,\ldots ,x_n)$ of the form
 \begin{subequations}
\begin{equation}\label{translate-t}
x\to x+c\rho_M,\quad  \rho_M=\sum_{1\leq j\leq n} (n-j)e_j
\end{equation}
and substituting $g_M\to g_M^{(c)}$ with
\begin{equation}\label{rescale-t}
g_M^{(c)}(g_M^{(c)}-1)=e^c \quad (g_M^{(c)}>0) ,
\end{equation}
\end{subequations}
the potential of $L_x^{\text{bc}}$ \eqref{Lbc} tends to that of $L_x^{\text{t}}$ \eqref{Lt}, \eqref{norm-coupling} for $c\to +\infty$.
 Similarly, the coordinate translation
  \begin{subequations}
\begin{equation}\label{translate-cs}
x\to x+c\rho_L,\quad  \rho_L=\sum_{1\leq j\leq n} e_j
\end{equation}
in $L_x^{\text{bc}}$ \eqref{Lbc} accompanied by the  parameter substitution $g_S\to g_S^{(c)}$, $g_L\to g_L^{(c)}$ with
\begin{equation}\label{rescale-cs}
 g_S^{(c)}=2g_S\quad\text{and}\quad    g_L^{(c)}(g_L^{(c)}-1)={\textstyle \frac{e^{2c}}{16}}\quad (g_L^{(c)}>0),
\end{equation}
\end{subequations}
leads one to recover $L_x^{\text{cs}}$ \eqref{Lcs}, \eqref{norm-coupling} for $c\to +\infty$. (Here the vectors $e_1,\ldots ,e_n$ refer to the standard unit basis of $\mathbb{C}^n$.)

The eigenfunctions of $L_x^{\text{bc}}$ \eqref{Lbc}  are moreover  known to be expressible in terms of the hyperoctahedral hypergeometric function of Heckman and Opdam associated with the root system $BC_n$  \cite{hec-sch:harmonic,opd:lecture}. Following a recent approach due to Oshima and Shimeno \cite{shi:limit,osh-shi:heckman-opdam}, we will infer
that the eigensolutions for $L_x^{\text{t}}$ \eqref{Lt} and $L_x^{\text{cs}}$ \eqref{Lcs} can be retrieved rigorously from this multivariate hypergeometric function
via partial confluences that amount to a careful lift of
the limit transitions $L_x^{\text{bc}}\to L_x^{\text{t}}$ and $L_x^{\text{bc}}\to L_x^{\text{cs}}$ to the level of the eigensolutions. For the special values of the spectral variable that parametrize the  eigenstates in the discrete spectrum, the hyperoctahedral hypergeometric series truncates and become multivariate Jacobi polynomials \cite{hec-sch:harmonic,opd:lecture,for:log-gases,hal-lan:unified}. The corresponding limit transition from these Jacobi polynomial eigenfunctions of $L_x^{\text{bc}}$ \eqref{Lbc} to the multivariate Bessel polynomial eigenfunctions of $ L_x^{\text{cs}}$ \eqref{Lcs} \cite{hal-lan:unified} was already used by Halln\"as to find the above-mentioned discrete spectrum eigenstates of
$L_x^{\text{cs}}$ \cite{hal:multivariable}. From this perspective, the  $\text{bc}\to \text{cs}$ confluence employed here should be viewed as an extension
of the limit transition in \cite{hal:multivariable}
to  the situation of arbitrary spectral values.
Somewhat different hypergeometric confluences have been considered in
Refs. \cite{jeu:paley-wiener,ros-voi:positivity,ben-ors:analysis,ben-ors:bessel,ros-voi:integral}, where  limit transitions from the hypergeometric eigenfunctions of the hyperbolic quantum Calogero-Sutherland systems to the multivariate Bessel eigenfunctions of the rational quantum Calogero systems were established, and  in Ref. \cite{ros-koo-voi:limit}  where a (dual) limit transition between the Heckman-Opdam hypergeometric functions of types $BC_n$ and $A_{n-1}$ was uncovered.
Moreover, in previous work   \cite{die-ems:bispectral} we retrieved a one-parameter family of hyperoctahedal Whittaker functions diagonalizing the quantum Toda chain with one-sided boundary perturbations of Morse type \cite{skl:boundary}  from the $BC_n$-type  Heckman-Opdam hypergeometric function, via a confluent limit that combines both the $\text{bc}\to \text{t}$ and the $\text{bc}\to \text{cs}$ confluences.
The latter Whittaker function unifies the ubiquitous $B_n$-type orthogonal Whittaker function (cf.  e.g. \cite{ish:whittaker,bis-zyg:point-to-line}) with the $C_n$-type symplectic Whittaker function, whereas---as emphasized above---here we arrive instead at a two-parameter hypergeometric deformation of the $D_n$-type orthogonal Whittaker function.

It is known that the Heckman-Opdam hypergeometric function and its confluent Whittaker simile satisfy
a system of bispectrally dual \cite{dui-gru:differential,gru:bispectral} difference equations in the spectral variable \cite{rui:finite-dimensional,che:inverse,cha:bispectrality,kha-leb:integral,bab:equations,hal-rui:kernel,skl:bispectrality,koz:aspects,bor-cor:macdonald,die-ems:difference}. These difference equations  can be interpreted themselves as eigenvalue equations  for a quantum integrable particle system characterized by rational Macdonald-Ruijsenaars type difference operators \cite{rui:complete,die:difference}. In other words, the multivariate special functions under consideration diagonalize these rational Macdonald-Ruijsenaars type quantum systems upon interchanging the role of the position and spectral variables (cf also \cite{rui:action-angle,rui:relativistic,feh-kli:duality,pus:hyperbolic,pus:scattering,feh:action-angle,feh-gor:duality,feh-mas:action-angle} for corresponding dualities within the realms of classical mechanics).
Fitting within this striking bispectral duality picture,  the partial hypergeometric confluences give rise to corresponding dual difference equations for the eigensolutions of $L_x^{\text{t}}$ \eqref{Lt} and $L_x^{\text{cs}}$ \eqref{Lcs}. With the aid of the  difference equations in question  we confirm that our eigensolutions for $L_x^{\text{t}}$ \eqref{Lt} and $L_x^{\text{cs}}$ \eqref{Lcs} are in fact holomorphic in the spectral variable.

The presentation is organized as follows. In Section \ref{sec2} the Harish-Chandra solutions of the eigenvalue equations for 
$L_x^{\text{t}}$ \eqref{Lt} and $L_x^{\text{cs}}$ \eqref{Lcs} are presented. In Section \ref{sec3} the $n$-particle wave functions for the corresponding
Schr\"odinger operators are constructed through
symmetrization of an appropriately normalized Harish-Chandra series with respect to the action of the hyperoctahedral group of signed permutations.
Section \ref{sec4} formulates a system of difference equations in the spectral variable  for the wave functions at issue, and 
in Section \ref{sec5} it is shown how these difference equations reveal that the wave functions are entire in the spectral variable. To facilitate the readability of this work, the construction of the Harish-Chandra series and most of the details regarding the
$\text{bc}\to \text{t}$ and $\text{bc}\to \text{cs}$  confluences at the level of the eigensolutions are relegated to two technical appendices at the end.

\section{Harish-Chandra series solution of the eigenvalue equation}\label{sec2}

For $\xi=(\xi_1,\ldots ,\xi_n)$, $x=(x_1,\ldots ,x_n)$ in $\mathbb{C}^n$, let
\begin{equation*}
\langle \xi ,x\rangle := \xi_1 x_1+\cdots +\xi_n x_n,
\end{equation*}
and let us write for $\nu=(\nu_1,\ldots ,\nu_n)\in\mathbb{Z}^n$  that
\begin{equation}\label{dominant}
\nu\geq 0 \Leftrightarrow \nu_1+\cdots +\nu_k\geq 0\ \text{for}\ k=1,\ldots ,n,
\end{equation}
and that $\nu >0$ if $\nu\geq 0$ and $\nu\neq 0$.

Given $\text{r}\in \{ \text{bc},\text{t},\text{cs}\}$,  we define a corresponding Harish-Chandra series
\begin{subequations}
\begin{equation}\label{HC:a}
\phi_\xi^{\text{r}}(x;g)= \sum_{ \nu\geq 0}  a^{\text{r}}_\nu (\xi ;g)  e^{\langle \xi -\nu ,x\rangle } ,
\end{equation}
with expansion coefficients $a^{\text{r}}_\nu (\xi ;g)$ determined by the recurrence
\begin{equation}\label{HC:b}
 \langle \nu-2\xi ,\nu\rangle a^{\text{r}}_\nu (\xi ; g) =   \sum_{\substack{\alpha\in R\\ l\geq 1}} \text{a}^{\text{r}}_{\alpha ,l}(g) a^{\text{r}}_{\nu-l\alpha}  (\xi ; g)\qquad (\nu >0) 
\end{equation}
in combination with
the initial condition
\begin{equation}\label{HC:c}
a^{\text{r}}_\nu (\xi ;g) := \begin{cases}   1&\text{if}\ \nu =0, \\ 0&\text{if}\ \nu\not\geq 0. \end{cases} 
\end{equation}
\end{subequations}
Here $g:=(g_S,g_M,g_L)\in\mathbb{C}^3$, $R:=R_S\cup R^+_M\cup R^-_M\cup R_L$ with
$R_S:=
\{ e_j  \mid 1\leq j \leq n\}$, $R^\pm_M:=  \{ e_j\pm e_k \mid  1\leq j < k\leq n\}$, $R_L:=
\{ 2e_j  \mid 1\leq j \leq n\}$, and
\begin{equation*}\label{a-alpha-bc}
\text{a}^{\text{bc}}_{\alpha ,l} (g) :=
\begin{cases}
g_S(g_S+2g_L-1) l&\text{if}\ \alpha \in R_S ,\\
2g_M(g_M-1) l&\text{if}\ \alpha\in R^+_M \cup R^-_M,\\
4g_L(g_L-1) l& \text{if}\ \alpha \in R_L  ,
\end{cases}
\end{equation*}
\begin{equation*}\label{a-alpha-t}
\text{a}^{\text{t}}_{\alpha ,l} (g) :=
\begin{cases}
g_S(g_S+2g_L-1) l&\text{if}\ \alpha=e_n ,\\
a_j \delta_{l-1} &\text{if}\ \alpha= e_j-e_{j+1}\   (1\leq j <n)    ,\\
a_{n-1}\delta_{l-1}  &\text{if}\ \alpha= e_{n-1}+e_{n},\\
4g_L(g_L-1) l& \text{if}\ \alpha =2e_n ,
\end{cases}
\end{equation*}
\begin{equation*}\label{a-alpha-cs}
\text{a}^{\text{cs}}_{\alpha ,l} (g) :=
\begin{cases}
g_S\delta_{l-1} &\text{if}\ \alpha \in R_S ,\\
2g_M(g_M-1) l&\text{if}\ \alpha\in R^-_M ,\\
a_n \delta_{l-1} & \text{if}\ \alpha \in R_L  
\end{cases}
\end{equation*}
(with the implicit assumption that
$ \text{a}^{\text{r}}_{\alpha ,l}(g) \equiv 0$ otherwise), while $\delta_l := 1$ if $l=0$ and $\delta_l:=0$ if $l\neq 0$.

\begin{proposition}[Harish-Chandra Series Solution]\label{HC:prp}
Let $\emph{r}\in \{ \emph{bc},\emph{t},\emph{cs}\}$.
 (i) The Harish-Chandra series $\phi_\xi^{\text{r}}(x;g)$ \eqref{HC:a}--\eqref{HC:c} constitutes an analytic function of $(\xi, x,g)\in \mathbb{C}^n_{+,\emph{reg}}\times\mathbb{A}^{\emph{r}}\times\mathbb{C}^3$, where
\begin{subequations}
\begin{equation}
\mathbb{C}_{\emph{reg},+}^n:=\{ \xi\in\mathbb{C}^n\mid 2\xi_j\not\in\mathbb{Z}_{> 0}, \,\xi_j\pm\xi_k \not\in\mathbb{Z}_{> 0}\, (j<k)\} 
\end{equation}
and
\begin{equation}\label{chamber}
\mathbb{A}^{\emph{r}}:= \begin{cases}  \{ x\in\mathbb{R}^n \mid x_1>x_2>\cdots >x_n>0\}  &\ \text{if}\ \emph{r}=\emph{bc} ,\\
 \{ x\in\mathbb{R}^n \mid x_n>0\}  &\ \text{if}\ \emph{r}=\emph{t} , \\
  \{ x\in\mathbb{R}^n \mid x_1>x_2>\cdots >x_n \}  &\ \text{if}\ \emph{r}=\emph{cs} ,
\end{cases} 
\end{equation}
\end{subequations}
which---as a (meromorphic)  function of the spectral variable $\xi\in\mathbb{C}^n$---has at most simple
poles along the hyperplanes belonging to $\mathbb{C}^n\setminus \mathbb{C}^n_{+,\emph{reg}}$.

(ii) The function in question provides an eigensolution for  $L_x^{\emph{r}}$ \eqref{Lt}, \eqref{Lcs}, \eqref{Lbc}:
\begin{subequations}
\begin{equation}\label{ef:eq}
L_x^{\emph{r}}\phi^{\emph{r}}_\xi (x;g) =\langle\xi,\xi\rangle \phi^{\emph{r}}_\xi (x;g)
\end{equation}
that enjoys a plane-wave asymptotics of the form
\begin{equation}\label{ef:as}
\lim_{x\to +\infty} | \phi^{\emph{r}}_{\xi }(x;g)- e^{\langle \xi, x\rangle} | =0 \quad\text{for}\quad \emph{Re}(\xi)=0,
\end{equation}
\end{subequations}
where the notation $x\to +\infty$ means that $x_k-x_{k+1}\to +\infty$ for $k=1,\ldots ,n$ (with the convention that $x_{n+1}\equiv 0$).
\end{proposition}

When $\text{r}=\text{bc}$, part (i) of this proposition is due to Opdam, cf. \cite[Cor. 2.2,\, Cor. 2.10]{opd:root} and \cite[Prp. 4.2.5]{hec-sch:harmonic}, while (the more straightforward) part (ii) is  already immediate from Opdam's previous work with Heckman \cite{hec-opd:root}, cf. also \cite[Sec. 4.2]{hec-sch:harmonic} and \cite[Sec. 6.1]{opd:lecture}.
For this case an elementary proof of part (i), based on the integrability of $L_x^{\text{bc}}$ in combination with a residue computation, was sketched in \cite[Lem. 6.5]{opd:lecture}.
In order to replicate the latter proof for the cases $\text{r}=\text{t}$ and $\text{r}=\text{cs}$, we would have to recur to the higher commuting quantum integrals from
Refs. \cite[Thm. 8]{osh:completely} and \cite[Sec. 5]{hal:multivariable}, respectively. 
Instead, we will avoid the explicit use of higher quantum integrals at this point and exploit---in Appendix \ref{appA} below---a Harish-Chandra type analysis in combination with partial confluent limit transitions $\phi_\xi^{\text{bc}}\to \phi_\xi^{\text{t}}$ and $\phi_\xi^{\text{bc}}\to \phi_\xi^{\text{cs}}$  in the spirit of  Oshima and Shimeno \cite{shi:limit,osh-shi:heckman-opdam}, so as to derive the corresponding statements for the cases $\text{r}=\text{t}$ and $\text{r}=\text{cs}$ directly from those for the known case $\text{r}=\text{bc}$.

\section{$n$-Particle wave function}\label{sec3}
For $(\xi ,x,g)\in  \mathbb{C}_{\text{reg}}^n\times\mathbb{A}^{\text{r}}\times\mathbb{C}^3$, where
 \begin{equation}\label{Creg}
 \mathbb{C}_{\text{reg}}^n:=\{ \xi\in\mathbb{C}^n\mid 2\xi_j\not\in\mathbb{Z}, \,\xi_j\pm\xi_k \not\in\mathbb{Z} \}  ,
 \end{equation}
we consider the $n$-particle wave function of the form
\begin{subequations}
\begin{equation}\label{wave-function:a}
\Phi_\xi^{\text{r}}(x;g):=\sum_{\text{w}\in W}  C^{\text{r}}(\text{w} \xi ;g )  \phi_{\text{w}\xi}^{\text{r}}(x;g) ,
\end{equation}
with
\begin{equation}\label{c-function}
C^{\text{r}}(\xi ;g ):=
 \prod_{1\leq j\leq n}    c_w (\xi_j)   \prod_{1\leq j<k\leq n} c_v(\xi_j+\xi_k)  c_v(\xi_j-\xi_k) ,
\end{equation}
\begin{equation}
c_v(z)=
\begin{cases}
\frac{\Gamma (z)}{\Gamma (g_M+z)} &\text{if}\ \text{r=bc}\ \text{or}\ \text{r=cs} ,\\
\Gamma (z) &\text{if}\ \text{r=t},
\end{cases}
\end{equation}
and
\begin{equation}\label{wave-function:d}
c_{w}(z)= \begin{cases}  \frac{\Gamma (2z ) \Gamma (\frac{1}{2}g_S+z)}{\Gamma (g_S+2z) \Gamma (\frac{1}{2}g_S+g_L+z)}&\text{if}\ \text{r=bc}\ \text{or}\ \text{r=t}, \\
\frac{\Gamma (2z)}{\Gamma (\frac{1}{2}+g_S+z)}    &\text{if}\ \text{r=cs}. 
\end{cases}
\end{equation}
\end{subequations}
Here $\Gamma (\cdot)$ refers to the gamma function and the
symmetrization in the spectral variable is with respect to the  action of the {\em hyperoctahedral group} $W:=  \{1 ,-1\}^n \rtimes S_n $
 of {\em signed permutations} $\text{w}=(\epsilon, \sigma )$ on $\mathbb{C}^n$:
\begin{equation}
\xi=(\xi_1,\ldots ,\xi_n)\stackrel{\text{w}}{\longrightarrow} (\epsilon_1 \xi_{\sigma^{-1} (1)},\ldots ,\epsilon_n \xi_{\sigma^{-1} (n)})=:\text{w}\xi ,
\end{equation}
where $\sigma = \left( \begin{matrix} 1& 2& \cdots & n \\
 \sigma (1)&\sigma (2)&\cdots & \sigma (n)
 \end{matrix}\right) $ belongs to the symmetric group $S_n$ and $\epsilon = (\epsilon_1,\ldots,\epsilon_n)$ with $\epsilon_j\in \{ 1,-1\}$ (for $j=1,\ldots, n$). It is clear from its definition
 and Proposition \ref{HC:prp} that  this wave function is analytic for $(\xi ,x,g)\in  \mathbb{C}_{\text{reg}}^n\times\mathbb{A}^{\text{r}}\times\mathbb{C}^3$.
The hyperoctahedral symmetry moreover guarantees that---as a meromorphic function of $\xi\in\mathbb{C}^n$---the simple poles
of  $ \Phi_\xi^{\text{r}}(x;g)$ \eqref{wave-function:a}--\eqref{wave-function:d} along the hyperplanes $\xi_j=0$ and $\xi_j\pm \xi_k=0$ (stemming from $C^{\text{r}}(\xi ;g )$) are removable. The following proposition is now immediate from Proposition \ref{HC:prp}.

\begin{proposition}[$n$-Particle Wave Function]\label{WF:prp} For any $\emph{r}\in \{ \emph{bc},\emph{t},\emph{cs}\}$,
the $n$-particle wave function $ \Phi_\xi^{\text{r}}(x;g)$ \eqref{wave-function:a}--\eqref{wave-function:d},  $(\xi ,x,g)\in  \mathbb{C}_{\emph{reg}}^n\times\mathbb{A}^{\emph{r}}\times\mathbb{C}^3$,
provides an analytic solution of the eigenvalue equation 
\begin{subequations}
\begin{equation}
L^{\emph{r}}_x\Phi^{\emph{r}}_\xi (x;g)  =\langle\xi,\xi\rangle \Phi^{\emph{r}}_\xi (x;g) 
\end{equation}
enjoying the following plane-wave asymptotics:
\begin{equation}
\lim_{x\to +\infty}  \Bigl|  \Phi^{\emph{r}}_{\xi }(x;g)-\sum_{\emph{w}\in W} C^{\emph{r}}(\emph{w}\xi;g) e^{\langle \emph{w}\xi, x\rangle}\Bigr| =0\quad\text{for}\quad\emph{Re}(\xi)=0.
\end{equation}
\end{subequations}
\end{proposition}

When $\text{r}=\text{bc}$, Proposition \ref{WF:prp} goes back to Heckman and Opdam \cite{hec-sch:harmonic,opd:lecture}. Indeed, the $n$-particle wave function
$\Phi_\xi^{\text{bc}}(x;g) $ \eqref{wave-function:a}--\eqref{wave-function:d} then amounts  in essence to their $BC_n$-type  hyperoctahedral hypergeometric function
$ F_{BC_n}(\xi,x;g)$ \cite{opd:root,hec-sch:harmonic,opd:lecture}:
\begin{subequations}
\begin{equation}
\Phi_\xi^{\text{bc}}(x;g) =  { \delta (x;g)}{C^{\text{bc}}(\rho_g ;g)}     F_{BC_n}(\xi,x;g)
\end{equation}
($(\xi,x,g)\in\mathbb{C}_{\text{reg}}^n\times \mathbb{A}^{\text{bc}}\times\mathbb{C}^3$), where
\begin{align}
\delta (x;g) :=& \prod_{1\leq j\leq n}   (e^{\frac{1}{2} x_j}  -  e^{-\frac{1}{2} x_j} )^{g_S} (e^{x_j}  -  e^{-x_j} )^{g_L}  \\
&\times \prod_{1\leq j <k\leq n}      (e^{\frac{1}{2} (x_j+x_k)}  -  e^{-\frac{1}{2} (x_j+x_k)} )^{g_M}   (e^{\frac{1}{2} (x_j-x_k)}  -  e^{-\frac{1}{2} (x_j-x_k)} )^{g_M}
\nonumber
\end{align}
and
\begin{equation}
\rho_g:={\textstyle ((n-1)g_M+\frac{1}{2}g_S+g_L,(n-2)g_M+\frac{1}{2}g_S+g_L,\ldots,\frac{1}{2}g_S+g_L)} .
\end{equation}
\end{subequations}

\section{Difference equations in the spectral variable}\label{sec4}
In \cite[Thm. 2]{die-ems:difference} a system of (bispectral dual) difference equations for the $BC_n$-type hypergeometric function of Heckman and Opdam was presented. The following theorem  extends this bispectral duality for the case $\text{r}=\text{bc}$ so as to include the cases $\text{r}=\text{t}$ and $\text{r}=\text{cs}$.

\begin{theorem}[Difference Equations]\label{WFD:thm}
For any $\emph{r}\in \{ \emph{bc},\emph{t},\emph{cs}\}$,
$\ell \in \{ 1,\ldots ,n\}$ and $(\xi,x,g)\in \mathbb{C}_{\emph{reg}}^n\times\mathbb{A}^{\emph{r}}\times\mathbb{C}^3$,
the wave function $\Phi^r_\xi(x;g)$ \eqref{wave-function:a}--\eqref{wave-function:d} satisfies the  following difference equation in the spectral variable:
\begin{subequations}
\begin{equation}\label{DE:a}
\sum_{\substack{J\subset \{ 1,\ldots ,n\} ,\, 0\leq|J|\leq \ell\\
               \epsilon_j \in \{ 1,-1\} ,\; j\in J}}
\!\!\!\!\!\!\!\!\!
U^{\emph{r}}_{J^c,\, \ell -|J|}(\xi ;g)
V^{\emph{r}}_{\epsilon J}(\xi ;g)
\Phi^{\emph{r}}_{\xi +e_{\epsilon J}} (x;g) =    E^{\emph{r}}_\ell (x)\Phi^{\emph{r}}_{\xi} (x;g) ,
\end{equation}
where
\begin{align}\label{V}
V^{\emph{r}}_{\epsilon J}(\xi ;g)&=
\prod_{j\in J} 
w (\epsilon_j\xi_j)
\prod_{\substack{j\in J\\ k\not\in J}} 
v(\epsilon_j\xi_j+\xi_k) v(\epsilon_j\xi_j-\xi_k) 
\nonumber \\
&  \times
\prod_{\substack{j,j^\prime \in J\\ j<j^\prime}}
v(\epsilon_j\xi_j+\epsilon_{j^\prime}\xi_{j^\prime})
v(\epsilon_j\xi_j+\epsilon_{j^\prime}\xi_{j^\prime}+1) ,
\end{align}
\begin{align}\label{U}
U^{\emph{r}}_{K,p}(\xi ;g)= (-1)^{p}
\sum_{\stackrel{I\subset K,\, |I|=p}
              {\epsilon_i  \in \{ 1,-1\} ,\; i\in I }}
&\Biggl( \prod_{i\in I} 
w(\epsilon_i\xi_i) \prod_{\substack{i\in I\\ k\in K\setminus I}} 
v(\epsilon_i\xi_i+\xi_{k})   v(\epsilon_i\xi_i-\xi_{k}) \\
&  \times
\prod_{\substack{i,i^\prime \in I\\ i<i^\prime}}
v(\epsilon_i\xi_i+\epsilon_{i^\prime}\xi_{i^\prime})
v(-\epsilon_i\xi_i-\epsilon_{i^\prime}\xi_{i^\prime}-1)\Biggr) ,
 \nonumber 
 \end{align}
 with
 \begin{equation}
v(z)= \begin{cases}    1+g_Mz^{-1} &\text{if}\ \emph{r=bc}\ \text{or}\ \emph{r=cs} ,\\
 z^{-1}  &\text{if}\ \emph{r=t},
\end{cases}
\end{equation}
\begin{equation}
w(z)= \begin{cases}   \left( 1+(\frac{1}{2}g_S+g_L)z^{-1} \right)  \left(   1+g_S(1+2z)^{-1}    \right)&\text{if}\ \emph{r=bc}\ \text{or}\ \emph{r=t}, \\
\frac{1}{2} z^{-1} \left(   \frac{1}{2}+g_S(1+2z)^{-1}    \right)  &\text{if}\ \emph{r=cs},
\end{cases}
\end{equation}
 and
 \begin{equation}\label{DE:f}
E^{\emph{r}}_\ell(x) =
\begin{cases}
{ 2^{2\ell }\sum_{\substack{J\subset \{ 1,\ldots, n\} \\ |J|=\ell}}   \prod_{j\in J}  \sinh^2\left(\frac{x_j}{2}\right) }&\text{if}\ \emph{r=bc} ,\\
e^{x_1+\cdots +x_\ell}    +\delta_{n-l}  e^{x_1+\cdots +x_{n-1}} (e^{-x_n}-2) &\text{if}\ \emph{r=t} ,\\
\sum_{\substack{J\subset \{ 1,\ldots, n\} \\ |J|=\ell}}   \prod_{j\in J}  e^{x_j}  &\text{if}\ \emph{r=cs} .
\end{cases}
\end{equation}
\end{subequations}
Here $|J|$ denotes the cardinality of $J\subset\{ 1,\ldots, n\}$,  $J^c:=\{ 1,\ldots, n\}\setminus J$, and 
$e_{\epsilon J} := \sum_{j\in J} \epsilon_j e_j $.
\end{theorem}

In Appendix \ref{appB} below we prove this theorem with the aid of the Oshima-Shimeno type confluences  $\Phi_\xi^{\text{bc}}\to \Phi_\xi^{\text{t}}$ and $\Phi_\xi^{\text{bc}}\to \Phi_\xi^{\text{cs}}$,  which arise 
by lifting the  transitions between the Harish-Chandra series
 $\phi_\xi^{\text{bc}}\to \phi_\xi^{\text{t}}$ and $\phi_\xi^{\text{bc}}\to \phi_\xi^{\text{cs}}$ from Appendix \ref{appA} to the level of the $n$-particle wave functions
$\Phi_\xi^{\text{r}}$ \eqref{wave-function:a}--\eqref{wave-function:d}.

The difference equations in Theorem \ref{WFD:thm}  reveal that---upon interchanging the role of the position variables and the spectral variables---the wave function $\Phi^{\text{r}}_\xi(x;g)$ \eqref{wave-function:a}--\eqref{wave-function:d} provides a joint eigenfunction
for the commuting families of rational Macdonald-Ruijsenaars type difference operators with hyperoctahedral symmetry introduced in \cite{die:difference}.
For $\ell =1,2$  the difference equations in question boil down to the two following identities:
\begin{subequations}
\begin{equation}\label{diffeqn1}
\sum_{\substack{ 1\leq j\leq n\\  \epsilon \in \{ 1,-1\} }}  V^{\text{r}}_{\epsilon j}(\xi;g) \Bigl( \Phi^{\text{r}}_{\xi +\epsilon e_j}(x;g)-\Phi^{\text{r}}_\xi (x;g)\Bigl) 
 = E^{\text{r}}_1(x) \Phi^{\text{r}}_\xi (x;g) 
\end{equation}
and 
\begin{align}\label{diffeqn2}
\sum_{\substack{1\leq j < j^\prime\leq n \\ \epsilon,\epsilon^\prime \in \{ 1, -1\} }} 
 &V^{\text{r}}_{\{   \epsilon j,\epsilon^\prime  j^\prime \} } (\xi ;g)  \left (\Phi^{\text{r}}_{\xi+\epsilon e_j+\epsilon^\prime e_{j^\prime}} (x;g)-\Phi^{\text{r}}_\xi (x;g)\right)
  + \\
  \sum_{1 \le j \le n,\epsilon\in\{1,-1\}}& U^{\text{r}}_{\{ 1,\ldots ,n\}\setminus \{ j\} , 1}(\xi ;g) V^{\text{r}}_{\epsilon j}(\xi ;g) \Phi^{\text{r}}_{\xi+\epsilon e_j} (x;g)
 = E^{\text{r}}_2(x)\Phi^{\text{r}}_\xi (x;g), \nonumber
\end{align}
\end{subequations}
respectively, where
 \begin{equation*}\label{HWD1:b}
V^{\text{r}}_{\epsilon j} (\xi;g)=w(\epsilon \xi_j)
\prod_{\substack{1\leq k\leq n\\ k\neq j}} v(\epsilon \xi_j+\xi_k)v(\epsilon\xi_j-\xi_k),
\end{equation*}
 \begin{align*}
 V^{\text{r}}_{\{   \epsilon j,\epsilon^\prime  j^\prime \} } (\xi ;g) =&
w(\epsilon \xi_j)
w(\epsilon^\prime \xi_{j^\prime})
v(\epsilon \xi_j+\epsilon^\prime\xi_{j^\prime})
v(1+\epsilon \xi_j+\epsilon^\prime\xi_{j^\prime})  \nonumber \\
&  \times \prod_{\substack{1\leq k\leq n \\k\neq j,j^\prime}} v(\epsilon \xi_j+\xi_k)  v(\epsilon \xi_j-\xi_k) ,
 \end{align*}
\begin{equation*}
U^{\text{r}}_{\{ 1,\ldots ,n\}\setminus \{ j\} , 1}(\xi ;g) = - \sum_{\substack{1\leq i\leq n,\, i \neq j\\ \epsilon\in\{ 1,-1\} }}
w(\epsilon \xi_i)
 \prod_{\substack{ 1\leq k\leq n \\ k\neq j , i} }
v(\epsilon\xi_i+\xi_{k})v(\epsilon\xi_i-\xi_{k}),
\end{equation*}
and
\begin{equation*}
E^{\text{r}}_1 (x) =
\begin{cases}
2^2\sum_{1\leq j\leq n}   \sinh^2 (\frac{x_j}{2}) &\text{if}\ \text{r=bc} ,\\
e^{x_1}    +\delta_{n-1}   (e^{-x_1}-2) &\text{if}\ \text{r=t} ,\\
\sum_{1\leq j\leq n}  e^{x_j}  &\text{if}\ \text{r=cs} ,
\end{cases}
\end{equation*}
\begin{equation*}
E^{\text{r}}_2 (x) =
\begin{cases}
2^4\sum_{1\leq j<k\leq n}   \sinh^2 (\frac{x_j}{2})  \sinh^2 (\frac{x_k}{2}) &\text{if}\ \text{r=bc} ,\\
e^{x_1+x_2}    +\delta_{n-2}   e^{x_1}(e^{-x_2}-2) &\text{if}\ \text{r=t} ,\\
\sum_{1\leq j<k\leq n}  e^{x_j+x_k}  &\text{if}\ \text{r=cs} .
\end{cases}
\end{equation*}
For $\text{r}=\text{bc}$, the  difference equation in Eq. \eqref{diffeqn1} (corresponding to $\ell=1$) is originally due to Chalykh \cite[Thm 6.12]{cha:bispectrality}.

\section{Analyticity in the spectral variable}\label{sec5}
A deep result of Opdam states that for $\text{r}=\text{bc}$ the $n$-particle wave function $\Phi^{\emph{r}}_\xi (x;g)$ \eqref{wave-function:a}--\eqref{wave-function:d} extends analytically  in $\xi$ to an entire function of
the spectral variable, cf.  \cite[Thm. 2.8]{opd:root}, \cite[Thm. 4.3.14]{hec-sch:harmonic} and \cite[Thm. 6.13]{opd:lecture}. 
While it seems nontrivial to adapt Opdam's proof to the cases $\text{r}=\text{t}$ and $\text{r}=\text{cs}$  when basing us on the tools currently at our disposal, we observe that 
 it is relatively straightforward to extend the result  to these two cases
via the difference equations in Eqs. \eqref{diffeqn1}, \eqref{diffeqn2}.

\begin{theorem}[Analyticity]\label{ac-wf:thm} For any $\emph{r}\in \{ \emph{bc},\emph{t},\emph{cs}\}$,
the wave function $\Phi^{\emph{r}}_\xi (x;g)$ \eqref{wave-function:a}--\eqref{wave-function:d} extends (uniquely) to an analytic function of
$(\xi,x,g)\in\mathbb{C}^n\times\mathbb{A}^{\emph{r}}\times\mathbb{C}^3$.
 \end{theorem}
 
 \begin{proof}
Since---viewed as a meromorphic function of the spectral variable $\xi\in\mathbb{C}^n$---the Harish-Chandra series $\phi^{\text{r}}_\xi (x;g)$ \eqref{HC:a}--\eqref{HC:c} exhibits at most simple poles along the hyperplanes of $\mathbb{C}^n\setminus\mathbb{C}^n_{+,\text{reg}}$ (by Proposition \ref{HC:prp}), and the gamma factors in  $c$-function $C^{\text{r}}(\xi ;g)$ \eqref{c-function} give at most rise to simple poles along the hyperplanes of $\mathbb{C}^n\setminus\mathbb{C}^n_{-,\text{reg}}$, where $\mathbb{C}^n_{-,\text{reg}}:=
\{ \xi\in\mathbb{C}^n\mid 2\xi_j\not\in\mathbb{Z}_{\leq 0}, \,\xi_j\pm\xi_k \not\in\mathbb{Z}_{\leq 0}\, (j<k)\} $, it is clear that
 the wave function
$\Phi^{\text{r}}_\xi (x;g)$  \eqref{wave-function:a}--\eqref{wave-function:d}  possesses at most simple poles along the hyperplanes of $\mathbb{C}^n\setminus\mathbb{C}^n_{\text{reg}}$.  
By Hartogs' theorem, it is enough to infer that the singularities in question are removable. Moreover, in view of the hyperoctahedral symmetry it thus suffices to
verify that  the residues
\begin{align*}
\text{Res}_{1;m}\Phi^{\text{r}}_\xi (x;g) &:= \lim_{2\xi_1\to m}   (2\xi_1-m)\Phi^{\text{r}}_\xi (x;g), \\
\text{Res}_{12;m} \Phi^{\text{r}}_\xi (x;g) &:= \lim_{\xi_1+\xi_2\to m}   (\xi_1+\xi_2-m)\Phi^{\text{r}}_\xi (x;g)
\end{align*}
on $\text{H}_{1;m}:=\{ \xi\in\mathbb{C}^n\mid  2\xi_1= m \} $ and $ \text{H}_{12;m}:=\{ \xi\in\mathbb{C}^n\mid  \xi_1+\xi_2 = m \} $ vanish for $m\in \mathbb{Z}_{\geq 0}$.
As observed just before Proposition \ref{WF:prp}, for $m=0$ the vanishing of the residues under consideration is  immediate from the hyperoctahedral symmetry. For $m>0$, we deduce the vanishing of the residues inductively by means of the difference equations \eqref{diffeqn1}, \eqref{diffeqn2}, upon
assuming that $\text{Res}_{1;k}\Phi^{\text{r}}_\xi(x;g)=0$ and $\text{Res}_{12;k}\Phi^{\text{r}}_\xi(x;g)=0$ for $0\leq k <m$.  
To this end Eqs.  \eqref{diffeqn1} and  \eqref{diffeqn2} are multiplied by appropriate powers  $(2\xi_1-m+2)^d$ and $(\xi_1+\xi_2-m+2)^d$, respectively,
before performing the corresponding limits
$2\xi_1-m+2\to 0$ and $\xi_1+\xi_2-m+2\to 0$. This entails simple linear
relations both for  $\text{Res}_{1;m}\Phi^{\text{r}}_\xi(x;g)$ and for  $\text{Res}_{12;m}\Phi^{\text{r}}_\xi(x;g)$,  from which it is manifest that the residues in question vanish.

\begin{table}[h]
{\footnotesize
\begin{equation*}
\begin{array}{ccll}
m             & d & \text{Linear relation stemming from Eq. \eqref{diffeqn1}}                                                            & \text{Additional relations} \\
& & & \\
1             & 2 &   \text{Res}_{1;-1}V_1 \text{Res}_{1;1}\Phi =0     &     \text{Res}_{1;-1}V_1   \not\equiv 0      \\
2             &   2&\text{Res}_{1;0}V_1 \text{Res}_{1;2} \Phi +
\text{Res}_{1;0}V_{-1} \text{Res}_{1;-2} \Phi  = 0  & \begin{cases} \text{Res}_{1;-2}   \Phi =-\text{Res}_{1;2}   \Phi  \\    \text{Res}_{1;0}V_{-1}=-\text{Res}_{1; 0}V_{1}     \not\equiv 0 \end{cases}  \\
3& 1&  \begin{array}{l}  [ V_1(\xi) ] _{2\xi_1=1} \text{Res}_{1;3}\Phi + \\
\ \ \text{Res}_{1;1} V_{-1} [ \Phi_{\xi-e_1}-\Phi_\xi]_{2\xi_1=1} = 0\end{array}   & \begin{cases}[ \Phi_{\xi-e_1}]_{2\xi_1=1}
=[\Phi_{\xi}]_{2\xi_1=1} \\    [ V_1(\xi) ] _{2\xi_1=1}  \not\equiv 0  
  \end{cases}     \\
  \ge 4 & 1&  [ V_1(\xi) ]_{ 2\xi_1= m-2} \text{Res}_{1;m} \Phi  = 0 &  [ V_1(\xi)] _{2\xi_1=m-2}   \not\equiv 0
\end{array}
\end{equation*}
}
\caption{Computation verifying that $\text{Res}_{1;m}\Phi^{\text{r}}_\xi(x;g)\equiv 0$} \label{res1:tab}
\end{table}

\begin{table}[h]
{\footnotesize
\begin{equation*}
\begin{array}{ccll}
m             & d & \text{Linear relation stemming from Eq. \eqref{diffeqn2}}                                                            & \text{Additional relations} \\
& & & \\
1             & 2 &  \text{Res}_{12;-1}V_{\{ +1,+2\} } \Bigl( \text{Res}_{12;1}\Phi 
-\text{Res}_{12;-1}\Phi  \Bigr) =0   &   \begin{cases}    \text{Res}_{12;-1}\Phi
=-\text{Res}_{12;1} \Phi  \\    \text{Res}_{12;-1}V_{\{ +1,+2\} }  \not\equiv 0          \end{cases}     \\
2             &   2&\begin{array}{l} \text{Res}_{12;0} V_{\{ +1,+2\} } \text{Res}_{12;2}\Phi
+\\
\ \ \text{Res}_{12;0} V_{\{ -1,-2\} } \text{Res}_{12;-2}\Phi
=0 \end{array} & \begin{cases}    \text{Res}_{12;2}\Phi
=-\text{Res}_{12;-2} \Phi   \\   \text{Res}_{12;0} V_{\{ +1,+2\} }
= \\ - \text{Res}_{12;0} V_{\{ -1,-2\} } \not\equiv 0 \end{cases}  \\
3& 1&   \begin{array}{l}  [ V_{\{ +1,+2\} }(\xi) ]_{\xi_1+\xi_2=1} \text{Res}_{12;3} \Phi  +\\
\ \ \text{Res}_{12;1} V_{\{ -1,-2\}}(\xi) [ \Phi_{\xi-e_1 -e_2}-\Phi_\xi ]_{\xi_1+\xi_2=1} 
=0 \end{array} & \begin{cases}   [ \Phi_{\xi-e_1 -e_2}]_{\xi_1+\xi_2=1}
= \\  [ \Phi_{\xi} ]_{\xi_1+\xi_2=1}   \\
 [ V_{\{ +1,+2\} }(\xi)  ]_{\xi_1+\xi_2=1}\not\equiv 0  \end{cases}     \\
  \ge 4 & 1& [ V_{\{ +1,+2\}}(\xi)]_{\xi_1+\xi_2=m-2}  \text{Res}_{12;m}\Phi = 0 & [ V_{\{ +1,+2\}}(\xi)]_{\xi_1+\xi_2=m-2}  \not\equiv 0
\end{array}
\end{equation*}
}
\caption{Computation verifying that $\text{Res}_{12;m}\Phi^{\text{r}}_\xi(x;g)\equiv 0$} \label{res12:tab}
\end{table}

For each case the details of the pertinent residue computation is encoded in Tables \ref{res1:tab} and \ref{res12:tab},
by specifying the corresponding linear relation for the residue stemming from the difference equation together with
some elementary additional relations  that permit to deduce the vanishing of the residue at issue. Notice that for typographical motives,  the superscript $\text{r}$ and the dependence on the variables are suppressed (whenever this does not lead to confusion).

For instance, in  order to check that $\text{Res}_{1;2}   \Phi^{\text{r}}_{\xi} (x;g) \equiv 0$  on $\text{H}_{1;2}\times\mathbb{A}^\text{r}\times\mathbb{C}^3$, one multiplies
Eq. \eqref{diffeqn1} by $4\xi_1^2$ before performing the limit 
$2\xi_1\to 0$, which entails the relation 
$
\text{Res}_{1;0}V^{\text{r}}_1(\xi ;g) \text{Res}_{1;0} \Phi^{\text{r}}_{\xi+e_1}(x;g) +
\text{Res}_{1;0}V^{\text{r}}_{-1}(\xi;g) \text{Res}_{1;0} \Phi^{\text{r}}_{\xi-e_1}(x;g) = 0.
$
Upon rewriting $
\text{Res}_{1;0}   \Phi^{\text{r}}_{\xi+e_1}(x;g)$ in terms of $\text{Res}_{1;2}   \Phi^{\text{r}}_{\xi}(x;g)
 $ and $
\text{Res}_{1;0}   \Phi^{\text{r}}_{\xi-e_1}(x;g)$ in terms of $\text{Res}_{1;-2}   \Phi^{\text{r}}_{\xi}(x;g)
 $, it follows that $\text{Res}_{1;2}   \Phi^{\text{r}}_{\xi} (x;g) \equiv 0$  on $\text{H}_{1;2}\times\mathbb{A}^{\text{r}}\times\mathbb{C}^3$, because
$
\text{Res}_{1;-2}   \Phi^{\text{r}}_{\xi}(x;g)=-\text{Res}_{1;2}   \Phi_{\xi}(x;g)
 $
by the hyperoctahedral symmetry, and $
\text{Res}_{1; 0}V^{\text{r}}_{-1} (\xi ;g) =-\text{Res}_{1;0}V^{\text{r}}_{1}(\xi ;g) 
      \not\equiv 0.
$
 \end{proof}

Notice that the arguments in the proof of Theorem \ref{ac-wf:thm} persist verbatim for $\text{r}=\text{bc}$. Still, this way  we do not arrive at a new alternative proof
of the analyticity in the spectral variable for the $BC_n$ hypergeometric function. The reason is that the corresponding $\text{r}=\text{bc}$ difference equations
in Theorem \ref{WFD:thm} were derived in Ref. \cite{die-ems:difference} from the Pieri  recurrence formulas for the $BC_n$ Jacobi polynomials by analytic continuation, i.e. by {\em invoking} Opdam's analyticity result in the first place. Combined with the present proof of Theorem \ref{ac-wf:thm}, this state of affairs reveals that---given the Pieri formulas for the $BC_n$ Jacobi polynomials (and growth estimates for the $BC_n$ hypergeometric function from \cite{opd:harmonic})---Opdam's analyticity in the spectral variable and
our $\text{r}=\text{bc}$ difference equations   in fact follow from each other.

\appendix

\section{Proof of Propostition \ref{HC:prp}}\label{appA}
To verify Proposition \ref{HC:prp}, a Harish-Chandra type analysis of Heckman and Opdam covering the case $\text{r}=\text{bc}$ \cite{opd:root,hec-sch:harmonic,opd:lecture} is broadened so as to incorporate the cases $\text{r}=\text{t}$ and $\text{r}=\text{cs}$. 
In order to confirm that the regularity domain of the  Harish-Chandra series  in the spectral variable for $\text{r}=\text{bc}$ \cite{opd:root} 
persists
when $\text{r}=\text{t}$ and $\text{r}=\text{cs}$, we employ
partial hypergeometric ($\text{bc}\to \text{t}$ and $\text{bc}\to \text{cs}$) confluences 
of a  type considered by Oshima and Shimeno
\cite{shi:limit,osh-shi:heckman-opdam}.

 \subsection{Local regularization of the Harish-Chandra coefficients}
Given  an arbitrarily chosen bounded  domain  $\text{U}\subset\mathbb{C}^n$ with closure $\overline{\text{U}}$, let 
\begin{equation}\label{regularization}
 \Delta_\text{U}(\xi): = \prod_{\substack{\mu>0 \\ \text{H}_\mu\cap \overline{\text{U}} \neq \emptyset}}   \langle \mu-2\xi,\mu\rangle ,
 \end{equation}
 where $\text{H}_\mu :=\{ \xi\in\mathbb{C}^n \mid  \langle \mu-2\xi,\mu\rangle =0\}$.
It is immediate from the recurrence relations for  $a^{\text{r}}_\nu (\xi;g)$ in Eqs. \eqref{HC:b}, \eqref{HC:c} that   for any $\nu\geq 0$ and any compact $\text{K}\subset \mathbb{C}^3$ the regularized Harish-Chandra coefficient
$\Delta_\text{U}(\xi) a^{\text{r}}_\nu (\xi;g)$  is holomorphic and bounded in $(\xi,g)$ on $\text{U}\times \text{K}$.

 \subsection{Convergence of the regularized Harish-Chandra series}
Let us
  deviate  in this subsection from the normalization convention in Eq. \eqref{norm-coupling} by allowing momentarily for an arbitrary configuration of auxiliary coupling constants
 $a:=(a_1,\ldots ,a_n)\in\mathbb{C}^n$. 
 Following
  \cite[Ch. IV, Lem. 5.3]{hel:groups} and \cite[Lem. 2.1]{opd:root}, one ratifies the subsequent uniform exponential bound on the growth of the recurrence coefficients.

 \begin{lemma}[Bound on the Recurrence Coefficients]\label{HC-bound:lem}
For any $\varepsilon >0$, there exists a constant $A>0$ (depending only on the above choice of $\emph{U}$, $\emph{K}$ and $\varepsilon$) such that for all $(\xi,g)\in \emph{U}\times \emph{K}$:
 \begin{equation}
 | \Delta_\emph{U}(\xi) a^{\emph{r}}_\nu (\xi;g)| \leq A e^{\varepsilon \langle \nu ,\rho\rangle}    
  \quad (\forall \nu \geq 0),
  \end{equation}
  where   $ \rho:= \rho_M+\rho_L=\sum_{1\leq j\leq n} (n+1-j) e_j$.
 \end{lemma}
 \begin{proof}
 Let us pick $\text{c}>0$ and $N>0$ such that for all  $\xi\in \text{U}$:
\begin{equation*}
\langle \nu ,\nu -2\xi\rangle \geq \text{c}\langle \nu ,\rho\rangle^2\qquad \forall\nu \geq 0\ \text{with}\  \langle \nu,\rho\rangle \geq N .
\end{equation*}
Given  $\varepsilon >0$, we fix  a positive integer $M\geq N$
such that $\forall g\in \text{K}$:
\begin{equation*}
\frac{1}{\text{c}} 
 \sum_{\substack{\alpha\in R \\ l\geq 1}}  e^{-\varepsilon l \langle \alpha ,\rho\rangle } |\text{a}^{\text{r}}_{\alpha ,l} (g) |
 \leq M.
\end{equation*}
Upon fixing $A>0$ sufficiently large such that $\forall (\xi,g)\in \text{U}\times \text{K}$:
$| \Delta_\text{U}(\xi) {a}^{\text{r}}_\nu (\xi ; g) |\leq A e^{\varepsilon \langle \nu ,\rho\rangle}$ for all $\nu\geq 0$ with $\langle\nu,\rho\rangle< M$,
one deduces inductively from the recurrence in Eq. \eqref{HC:b} that also for all $\nu\geq 0$ such that $\langle\nu,\rho\rangle\geq M$:
\begin{align*}
| \Delta_\text{U}(\xi) {a}^{\text{r}}_\nu (\xi ; g) | &\leq
\frac{1}{\text{c} \langle \nu ,\rho\rangle^2}
  \sum_{\substack{\alpha\in R \\ l\geq 1}}  | \text{a}^{\text{r}}_{\alpha ,l} ( g)  \Delta_\text{U}(\xi) {a}^{\text{r}}_{\nu -l\alpha}(\xi ; g) |   \\
&\leq
\frac{1}{\text{c} \langle \nu ,\rho\rangle^2}
  \sum_{\substack{\alpha\in R \\ l\geq 1}}   | \text{a}^{\text{r}}_{\alpha ,l} (g) |   A e^{\varepsilon \langle \nu-l\alpha ,\rho\rangle}
  \\
& \leq {\textstyle \frac{M}{ \langle \nu ,\rho\rangle^2}} A e^{\varepsilon \langle \nu ,\rho\rangle} \leq A e^{\varepsilon \langle \nu ,\rho\rangle} .
\end{align*}
 \end{proof}

We see from Lemma  \ref{HC-bound:lem} that for any $\epsilon>0$
the series $e^{-\langle \xi ,x\rangle}\Delta_\text{U}(\xi )\phi^{\text{r}}_\xi(x;g)$ \eqref{HC:a}--\eqref{HC:c}
converges uniformly in absolute value when $(\xi,x,g)$ remains restricted to $ \text{U}\times (\mathbb{A}^{\text{bc}}+\epsilon \rho )\times \text{K}$. Indeed, one has that
\begin{equation}\label{bounded-convergence}
 \sum_{ \nu\geq 0}  |\Delta_\text{U}(\xi) {a}^{\text{r}}_\nu (\xi ; g)  e^{-\langle  \nu ,x  \rangle } |
\leq A   \sum_{ \nu\geq 0}   e^{- \langle \nu ,x-\varepsilon\rho \rangle } ,
\end{equation}
and the series on the RHS converges uniformly in $x$ on $ (\mathbb{A}^{\text{bc}}+\epsilon \rho )$ provided $0<\varepsilon<\epsilon$.

The upshot is that the regularized Harish-Chandra series $\Delta_\text{U}(\xi )\phi^{\text{r}}_\xi(x;g)$ converges   locally uniformly in absolute value for
$(\xi,x,g)$ in $ \text{U}\times \mathbb{A}^{\text{bc}} \times \mathbb{C}^3$. To infer that for $\text{r}=\text{t}$ and $\text{r}=\text{cs}$ this locally uniform  absolute convergence  extends in fact to the larger domain  $ \text{U}\times \mathbb{A}^{\text{r}} \times \mathbb{C}^3$, it suffices to observe that the translation \eqref{translate-t} amounts in the series $e^{-\langle \xi ,x\rangle}\phi^{\text{t}}_\xi (x;g)$ to the parameter shift $a_j\to a_j e^{-c}$ ($j=1,\ldots ,n-1$) and that
the translation  \eqref{translate-cs} amounts  in the series $e^{-\langle \xi ,x\rangle}\phi^{\text{cs}}_\xi (x;g)$
to the parameter shift $g_S\to g_Se^{-c}$,  $a_n\to a_n e^{-2c}$ ($c\in\mathbb{R}$). Notice in this connection that
 $\mathbb{A}^{\text{t}}= \cup_{c\in \mathbb{R}}  (\mathbb{A}^{\text{bc}}+c\rho_M)$ and $\mathbb{A}^{\text{cs}}= \cup_{c\in \mathbb{R}}  (\mathbb{A}^{\text{bc}}+c\rho_L)$.

\subsection{Eigenvalue equation}
By acting with the Schr\"odinger operator $L_x^{\text{r}}$ on the regularized Harish-Chandra series $\Delta_\text{U}(\xi )\phi^{\text{r}}_\xi(x;g)$, one sees that for  $(\xi,x,g)\in \text{U}\times \mathbb{A}^{\text{r}} \times \mathbb{C}^3$:
\begin{align*}
L_x^{\text{r}} \bigl( \Delta_\text{U}(\xi ) & \phi^{\text{r}}_\xi (x;g) \bigr) =  \\
&\sum_{\nu \geq 0} \Delta_\text{U}(\xi ) a_\nu^{\text{r}} (\xi;g) e^{\langle \xi -\nu,x\rangle}
\Bigl(
\langle \xi -\nu,\xi -\nu \rangle   
- \sum_{\substack{\alpha\in R\\ l\geq 1}} \text{a}^{\text{r}}_{\alpha ,l}(g) e^{- l \langle \alpha ,x\rangle} \Bigr) .
\end{align*}
Here the $\sinh^{-2}(\cdot) $ potentials were expanded using series
$\sinh^{-2}(z)= \sum_{l\geq 1} 4l e^{-2lz}$ ($\text{Re}(z)>0$), and the termwise
differentiation is justified by the locally uniform convergence in absolute value of the Harish-Chandra series. 
Rearrangement of the terms and invoking of the recurrence relations \eqref{HC:b}, \eqref{HC:c} produces
\begin{align*}
& \Delta_\text{U}(\xi ) \sum_{\nu \geq 0} e^{\langle \xi -\nu,x\rangle}
\Bigl(   \bigl( \langle\xi ,\xi\rangle +\langle \nu,\nu-2\xi\rangle \bigr)
  a^{\text{r}}_\nu (\xi;g) -
 \sum_{\substack{\alpha\in R\\ l\geq 1}} \text{a}^{\text{r}}_{\alpha ,l}(g) a^{\text{r}}_{\nu-l\alpha}  (\xi ; g)  \Bigr) \\
&= \langle\xi ,\xi\rangle
 \Delta_\text{U}(\xi )  \sum_{\nu \geq 0}  a^{\text{r}}_\nu (\xi;g) e^{\langle \xi -\nu,x\rangle}= \langle\xi ,\xi\rangle  \Delta_\text{U}(\xi )    \phi^{\text{r}}_\xi (x;g) .
\end{align*}

\subsection{Plane-wave asymptotics}
For any bounded domain $\text{U}\subset \imath \mathbb{R}^n$ (where $\imath   :=\sqrt{-1}$), the product $\Delta_\text{U}(\xi)$ \eqref{regularization} becomes empty, i.e.  $\Delta_\text{U}(\xi)\equiv 1$ in this situation.
The bound in Lemma \ref{HC-bound:lem} now reveals that for any $\xi\in \imath \mathbb{R}^n$:
\begin{align*}
&\lim_{x\to +\infty} | \phi^{\text{r}}_\xi (x;g)-e^{\langle \xi ,x\rangle} |  \leq  \lim_{x\to +\infty}   \sum_{\nu > 0}  | a^{\text{r}}_\nu (\xi;g) e^{-\langle \nu,x\rangle} |
\\
& \leq   
 \lim_{x\to +\infty}  A  \sum_{\nu > 0} e^{-\langle \nu,x-\varepsilon\rho \rangle}
  =    A \sum_{\nu > 0} \Bigl(  \lim_{x\to +\infty}e^{-\langle \nu,x-\varepsilon\rho \rangle} \Bigr) =0
\end{align*}
(by dominated convergence). Here we used that for $\nu >0$: 
 $ \lim_{x\to +\infty} e^{-\langle \nu,x-\varepsilon\rho\rangle}=0$ and $e^{-\langle \nu,x-\varepsilon\rho \rangle} < 1$ if $x\in (\mathbb{A}^{\text{bc}}+\varepsilon\rho )$.

\subsection{Partial hypergeometric confluences}
 We now lift the above $L_x^\text{bc}\to L_x^\text{t}$ and $L_x^\text{bc}\to L_x^\text{cs}$ limit transitions to the level of the Harish-Chandra series, following a strategy developed by Oshima and Shimeno \cite{shi:limit,osh-shi:heckman-opdam}.

\begin{proposition}[Partial Confluent Limits of the Harish-Chandra Series]\label{HC-confluence:prp}
Let $(\xi,x,g) \in \emph{U}\times\mathbb{A}^{\emph{bc}}\times\mathbb{C}^3$.

(i) The translation $x\to x+c\rho_M$ \eqref{translate-t} and the parameter substitution $g\to g^{(c)}$ \eqref{rescale-t} give rise to the $\emph{bc}\to \emph{t}$ confluence
\begin{subequations}
\begin{equation}
\lim_{c\to +\infty}   e^{-c\langle \xi ,\rho_M\rangle} \Delta_\emph{U}(\xi) \phi_\xi^{\emph{bc}}(x+c\rho_M; g^{(c)})= \Delta_\emph{U}(\xi) \phi^{\emph{t}}_\xi (x;g) .
\end{equation}

(ii) The translation $x\to x+c\rho_L$ \eqref{translate-cs} and the parameter substitution $g\to g^{(c)}$ \eqref{rescale-cs} give rise to the $\emph{bc}\to \emph{cs}$ confluence\begin{equation}
\lim_{c\to +\infty}   e^{-c\langle \xi ,\rho_L\rangle} \Delta_\emph{U}(\xi) \phi_\xi^{\emph{bc}}(x+c\rho_L; g^{(c)})= \Delta_\emph{U}(\xi) \phi^{\emph{cs}}_\xi (x;g) .
\end{equation}
\end{subequations}
\end{proposition}

\begin{proof}
The substitutions
\eqref{translate-t}, \eqref{rescale-t}  and  \eqref{translate-cs}, \eqref{rescale-cs}  in
the Harish-Chandra series $ \phi_\xi^{\text{bc}}(x;g) $ \eqref{HC:a}--\eqref{HC:c} reveal that
\begin{equation}\label{series}
 e^{-c\langle \xi ,\rho_K\rangle} \Delta_\text{U}(\xi) \phi_\xi^{\text{bc}}(x+c\rho_K; g^{(c)})= \sum_{ \nu\geq 0}  \hat{a}^{\text{bc}}_{\nu } (\xi ; g^{(c)})  e^{\langle \xi -\nu ,x\rangle } ,
\end{equation}
with $\hat{a}^{\text{bc}}_\nu (\xi ; g^{(c)})$, $\nu\geq 0$ given by
\begin{subequations}
\begin{equation}\label{rec:a}
 \langle \nu-2\xi ,\nu\rangle \hat{a}^{\text{bc}}_\nu (\xi ; g^{(c)}) =   \sum_{\substack{\alpha\in R\\ l\geq 1}} e^{-c\, l \langle \alpha ,\rho_K\rangle } \text{a}^{\text{bc}}_{\alpha ,l}(g^{(c)}) \hat{a}^{\text{bc}}_{\nu-l\alpha}  (\xi ;g^{(c)})\qquad (\nu >0) 
\end{equation}
and
\begin{equation}\label{rec:b}
\hat{a}^{\text{bc}}_\nu (\xi ;g^{(c)}) := \begin{cases} \Delta_\text{U}(\xi)  &\text{if}\ \nu =0, \\ 0&\text{if}\ \nu\not\geq 0, \end{cases} 
\end{equation}
\end{subequations}
where $K=M$ or $K=L$, respectively.

Since the sum on the RHS of Eq.  \eqref{rec:a} is finite (as $l\leq \langle \nu,\rho\rangle$ for $\nu-l\alpha\geq 0$) and
\begin{equation}\label{a-cs-t-lim}
\lim_{c\to +\infty}  e^{-c\, l \langle \alpha ,\rho_K\rangle } \text{a}^{\text{bc}}_{\alpha ,l} (g^{(c)})=\begin{cases}
\text{a}^{\text{t}}_{\alpha ,l}(g)& \text{if}\  K=M ,\\
\text{a}^{\text{cs}}_{\alpha ,l}(g)& \text{if}\  K=L,
\end{cases}
\end{equation}
it follows inductively in $\nu$ from the recurrence relations that  for all $ \xi\in \text{U}$ and $\nu \geq 0$:
\begin{equation}\label{coef-lim}
\lim_{c\to +\infty}   \hat{a}^{\text{bc}}_\nu (\xi ; g^{(c)}) =\begin{cases}
\Delta_\text{U}(\xi) a^{\text{t}}_\nu (\xi ;g) & \text{if}\  K=M ,\\
\Delta_\text{U}(\xi) a^{\text{cs}}_\nu (\xi ;g) & \text{if}\  K=L .
\end{cases}
\end{equation}
With the aid of an exponential bound on the growth of the recurrence coefficients $ \hat{a}^{\text{bc}}_\nu (\xi ; g^{(c)})$
similar to that of Lemma \ref{HC-bound:lem}, the asserted limit transitions now follow from Eqs. \eqref{series} and \eqref{coef-lim} by dominated convergence.
Indeed---by mimicking the proof of Lemma \ref{HC-bound:lem} and using Eq. \eqref{coef-lim}---it is readily seen
that  for any $\varepsilon >0$ there exists a constant $A>0$ (depending only on  $\xi\in \text{U}$, $g\in\mathbb{C}^3$ and $\varepsilon$)  such that
\begin{equation}\label{HC-cs:estimate}
\forall  c\geq 0, \nu\geq 0: \quad | \hat{a}^{\text{bc}}_\nu (\xi ; g^{(c)})| \leq A\,  e^{\varepsilon \langle \nu ,\rho\rangle} .
\end{equation}
Hence---after picking $\varepsilon >0$ sufficiently small such that $x\in (\mathbb{A}^{\text{bc}}+\varepsilon\rho )$---the
series on the RHS of Eq. \eqref{series} is termwise bounded uniformly in $c\geq 0$ by the following absolutely convergent series (cf.  Eq. \eqref{bounded-convergence}):
\begin{equation*}
 \sum_{ \nu\geq 0}  | \hat{a}^{\text{bc}}_\nu (\xi ; g^{(c)})  e^{\langle \xi -\nu ,x  \rangle } |
\leq A   | e^{\langle \xi ,x\rangle } | \sum_{ \nu\geq 0}   e^{- \langle \nu ,x-\varepsilon\rho \rangle }< +\infty .
\end{equation*}
\end{proof}

\subsection{Regularity domain}
The  locally uniform convergence of the regularized Harish-Chandra series guarantees that $\Delta_\text{U}(\xi) \phi_\xi^{\text{r}}(x;g)$ is analytic for $(\xi ,x, g)\in \text{U}\times\mathbb{A}^{\text{r}}\times \mathbb{C}^3$.  Moreover, it was shown by Opdam (cf. 
 \cite[Cor. 2.2,\, Cor. 2.10]{opd:root}, \cite[Prp. 4.2.5]{hec-sch:harmonic} and \cite[Lem. 6.5]{opd:lecture}) that the series
 $ \phi_\xi^{\text{r}}(x;g)$ is analytic for $(\xi ,x, g)\in \mathbb{C}^n_{\text{reg},+}\times\mathbb{A}^{\text{r}}\times \mathbb{C}^3$ when $\text{r}=\text{bc}$. In other words,  in this situation
 $\Delta_\text{U}(\xi) \phi_\xi^{\text{r}}(x;g)$ vanishes  for  $(\xi ,x, g)\in (\text{U}\cap \text{H}_\mu )\times\mathbb{A}^{\text{r}}\times \mathbb{C}^3$ if
the hyperplane $\text{H}_\mu $, $\mu >0$ does not belong to $ \mathbb{C}^n\setminus \mathbb{C}^n_{\text{reg},+}$.
 By Proposition \ref{HC-confluence:prp}, the vanishing property in question persists when $\text{r}=\text{t}$ and $\text{r}=\text{cs}$ (where we use the analyticity in $x$
 to extend from $\mathbb{A}^{\text{bc}}$ to $\mathbb{A}^{\text{r}}$). 
As the bounded domain $\text{U}\subset\mathbb{C}^n$ was chosen arbitrarily, the upshot is that also for $\text{r}=\text{t}$ and $\text{r}=\text{cs}$  the series
 $ \phi_\xi^{\text{r}}(x;g)$ must be analytic for $(\xi ,x, g)\in \mathbb{C}^n_{\text{reg},+}\times\mathbb{A}^{\text{r}}\times \mathbb{C}^3$,
with at most a simple pole arising along $\text{H}_\mu \times\mathbb{A}^{\text{r}}\times \mathbb{C}^3$ if the hyperplane $\text{H}_\mu $, $\mu >0$ belongs to $ \mathbb{C}^n\setminus \mathbb{C}^n_{\text{reg},+}$.

\section{Proof of Theorem \ref{WFD:thm}}\label{appB}
Starting from the difference equations \eqref{DE:a}--\eqref{DE:f} for the case $\text{r}=\text{bc}$ from \cite[Thm. 2]{die-ems:difference}, the cases $\text{r}=\text{t}$ and $\text{r}=\text{cs}$ are recovered via the following $\text{bc}\to \text{t}$ and $\text{bc}\to \text{cs}$ confluences
that extend Proposition \ref{HC-confluence:prp} to the level of the $n$-particle wave functions
$ \Phi_\xi^{\emph{r}}(x; g )$ \eqref{wave-function:a}--\eqref{wave-function:d} (cf. \cite{shi:limit,osh-shi:heckman-opdam}).

\begin{proposition}[Partial Confluent Limits of the Wave Function]\label{WF-confluence:prp}
Let $(\xi,x,g) \in \mathbb{C}_{\emph{reg}}^n\times\mathbb{A}^{\emph{bc}}\times\mathbb{C}^3$.

(i) The translation $x\to x+c\rho_M$ \eqref{translate-t} and the parameter substitution $g\to g^{(c)}$ \eqref{rescale-t} give rise to the  $\emph{bc}\to \emph{t}$ confluence
\begin{subequations}
\begin{equation}
\lim_{c\to +\infty}    \gamma_M (g^{(c)})  \Phi_\xi^{\emph{bc}}(x+c\rho_M; g^{(c)})=\Phi^{\emph{t}}_\xi (x;g) ,
\end{equation}
where
\begin{equation*}
  \gamma_M (g) :=   \left( \Gamma (g_M)    \right)^{n(n-1)}  .
\end{equation*}

(ii) The translation $x\to x+c\rho_L$ \eqref{translate-cs} and the parameter substitution $g\to g^{(c)}$ \eqref{rescale-cs} give rise to the $\emph{bc}\to \emph{cs}$  confluence
\begin{equation}
\lim_{c\to +\infty}    \gamma_L (g^{(c)})\Phi_\xi^{\emph{bc}}(x+c\rho_L; g^{(c)})= \Phi^{\emph{cs}}_\xi (x;g) ,
\end{equation}
where
\begin{equation*}
  \gamma_L (g) :=    \left(  \frac{\Gamma (g_S)\Gamma (\frac{1}{2}g_S+g_L)}{\Gamma (\frac{1}{2} g_S)\Gamma (\frac{1}{2}+\frac{1}{2}g_S)} \right)^n .
\end{equation*}
\end{subequations}
\end{proposition}

\begin{proof}
From the  well-known asymptotics
$
 \lim_{c\to +\infty}   \frac{c^z \Gamma (c)  }{\Gamma (z+c)} =1
$
and the duplication formula $
 \Gamma (2z)=\pi^{-\frac{1}{2}} 2^{2z-1}\Gamma (z)\Gamma \textstyle{(\frac{1}{2}}+z)
$
($2z\not\in\mathbb{Z}_{\leq 0}$), it is immediate that upon substituting $g\to g^{(c)}$ \eqref{rescale-t}:
\begin{subequations}
\begin{equation}\label{c-function-lim:a}
\lim_{c\to +\infty}   \gamma_M (g^{(c)}) \,e^{c\langle \xi ,\rho_M\rangle }  C^{\text{bc}}(\xi;g^{(c)}) =C^{\text{t}} (\xi;g) ,
\end{equation}
and upon substituting $g\to g^{(c)}$ \eqref{rescale-cs}:
\begin{equation}\label{c-function-lim:b}
\lim_{c\to +\infty}   \gamma_L (g^{(c)}) \,e^{c\langle \xi ,\rho_L\rangle }  C^{\text{bc}}(\xi;g^{(c)}) =C^{\text{cs}} (\xi;g) .
\end{equation}
\end{subequations}
The asserted limit transitions now follow directly from the explicit formulas in Eqs. \eqref{wave-function:a}--\eqref{wave-function:d}, with aid of Proposition \ref{HC-confluence:prp} and Eqs. \eqref{c-function-lim:a}, \eqref{c-function-lim:b}.
 \end{proof}

(i) To retrieve the difference equations for $\text{r}=\text{t}$, one
effectuates the translation \eqref{translate-t} and the parameter substitution  \eqref{rescale-t} in
the $\ell$th difference equation for $\text{r}=\text{bc}$. Upon multiplication by an overall 
scaling factor of the form $e^{-\frac{c}{2}\ell (2n-1-\ell)} \gamma_M(g^{(c)})$, this produces the corresponding difference equation of type $\text{r}=\text{t}$ for $c\to +\infty$.
Here one uses Propostion \ref{WF-confluence:prp} and the elementary limits
\begin{align*}
\lim_{c\to +\infty}  e^{-\frac{c}{2}|J| (2n-1-|J|)} V^{\text{bc}}_{\varepsilon J}(\xi ;g^{(c)})&=V^{\text{t}}_{\varepsilon J}(\xi ;g ),  \\
\lim_{c\to +\infty}  e^{-\frac{c}{2} p (2|K|-1-p)}  U^{\text{bc}}_{K,p}(\xi ;g^{(c)}) &= U^{\text{t}}_{K,p}(\xi ;g ) ,\\
\lim_{c\to +\infty}  e^{-\frac{c}{2}\ell (2n-1-\ell)}   E^{\text{bc}}_\ell (x+c\rho_M)&=  E^{\text{t}}_\ell (x),
\end{align*}
with $|K|=n-|J|$, $p=\ell-|J|$.

(ii) To retrieve the difference equations for $\text{r}=\text{cs}$, one
effectuates the translation \eqref{translate-cs} and the parameter substitution  \eqref{rescale-cs} in
the $\ell$th difference equation for $\text{r}=\text{bc}$. Upon multiplication by an overall 
scaling factor of the form $ e^{-c\ell} \gamma_L(g^{(c)})$, this produces the corresponding difference equation of type $\text{r}=\text{cs}$ for $c\to +\infty$.
Here one uses Propostion \ref{WF-confluence:prp} and the elementary limits
\begin{align*}
\lim_{c\to +\infty}  e^{-c |J|} V^{\text{bc}}_{\varepsilon J}(\xi ;g^{(c)})&=V^{\text{cs}}_{\varepsilon J}(\xi ;g),  \\
\lim_{c\to +\infty}  e^{-cp}  U^{\text{bc}}_{K,p}(\xi ;g^{(c)}) &= U^{\text{cs}}_{K,p}(\xi ;g) ,\\
 \lim_{c\to +\infty}  e^{-c\ell}   E^{\text{bc}}_\ell (x+c\rho_L)&=  E^{\text{cs}}_\ell (x) ,
\end{align*}
with $|K|=n-|J|$, $p=\ell-|J|$.

In both cases one initially recovers the difference equations for the restricted domain of $(\xi,x,g) \in \mathbb{C}^n_{\text{reg}}\times\mathbb{A}^{\text{bc}}\times\mathbb{C}^3$. The extensions to $(\xi,x,g)$ in $\mathbb{C}^n_{\text{reg}}\times\mathbb{A}^{\text{t}}\times\mathbb{C}^3$ and
$\mathbb{C}^n_{\text{reg}}\times\mathbb{A}^{\text{cs}}\times\mathbb{C}^3$ are plain by analytic continuation.

\bibliographystyle{amsplain}

\begin{thebibliography}{0000000}

\bibitem[A77]{adl:some} M. Adler, Some finite dimensional integrable systems and their scattering behavior, Comm. Math. Phys. {\bf 55} (1977),
195--230. 

\bibitem[B03]{bab:equations} O. Babelon,
Equations in dual variables for Whittaker functions,
Lett. Math. Phys. {\bf 65} (2003), 229--240. 

\bibitem[BO11]{bau-oco:exponential} F. Baudoin and N. O'Connell,
Exponential functionals of Brownian motion and class-one Whittaker functions,
Ann. Inst. Henri Poincar\'e Probab. Stat. {\bf 47} (2011), 1096--1120. 

\bibitem[BO05a]{ben-ors:analysis} S. Ben Sa\"{\i}d and B.  \O{}rsted, Analysis on flat symmetric spaces, J. Math. Pures Appl. {\bf 84} (2005), 1393--1426. 

\bibitem[BO05b]{ben-ors:bessel} \bysame, Bessel functions for root systems via the trigonometric setting, Int. Math. Res. Not. IMRN {\bf 2005}(9), 551--585.

\bibitem[BZ17]{bis-zyg:point-to-line}  E. Bisi and N. Zygouras, Point-to-line polymers and orthogonal Whittaker functions, Trans. Amer. Math. Soc. (to appear).

\bibitem[BC14]{bor-cor:macdonald} A. Borodin and I. Corwin, Macdonald processes, Probab. Theory Related Fields {\bf 158} (2014), 225--400.

\bibitem[C00]{cha:bispectrality} O.A. Chalykh,
Bispectrality for the quantum Ruijsenaars model and its integrable deformation,
J. Math. Phys. {\bf 41} (2000), 5139--5167. 

\bibitem[C97]{che:inverse} I. Cherednik, Inverse Harish-Chandra transform and difference operators, Internat. Math. Res. Not. IMRN {\bf 1997}(15), 733--750. 

\bibitem[D95]{die:difference} J.F. van Diejen, Difference Calogero-Moser systems and finite Toda chains,
J. Math. Phys. {\bf 36} (1995), 1299--1323.

\bibitem[DE15]{die-ems:difference}  J.F. van Diejen and E. Emsiz, Difference equation for the Heckman-Opdam hypergeometric function and its confluent Whittaker limit, Adv. Math. {\bf 285} (2015), 1225--1240.

\bibitem[DE18]{die-ems:bispectral}  \bysame, Bispectral dual difference equations for the quantum Toda chain with
boundary perturbations, Internat. Math. Res. Not. IMRN, DOI: 10.1093/imrn/rnx219

\bibitem[DG86]{dui-gru:differential} J.J. Duistermaat and F.A. Gr\"unbaum,
Differential equations in the spectral parameter,
Comm. Math. Phys. {\bf 103} (1986), 177--240. 

\bibitem[F13]{feh:action-angle} L. Feh\'er,
Action-angle map and duality for the open Toda lattice in the perspective of Hamiltonian reduction,
Phys. Lett. A {\bf 377} (2013), 2917--2921. 

\bibitem[FG14]{feh-gor:duality} L. Feh\'er and T.F. G\"orbe, Duality between the trigonometric $BC_n$ Sutherland system and a completed rational Ruijsenaars-Schneider-van Diejen system, J. Math. Phys. {\bf 55}  (2014), 102704.

\bibitem[FK09]{feh-kli:duality} L. Feh\'er and C. Klim{\v{c}}{\'{\i}}k,
On the duality between the hyperbolic Sutherland and the rational Ruijsenaars-Schneider models,
J. Phys. A {\bf 42} (2009), 185202. 

\bibitem[FM17]{feh-mas:action-angle} L. Feh\'er and I. Marshall, The action-angle dual of an integrable Hamiltonian system of
Ruijsenaars-Schneider-van Diejen type, J. Phys. A {\bf 50} (2017), 314004.

\bibitem[F10]{for:log-gases} P.J. Forrester, Log-gases and Random Matrices, London Mathematical Society Monographs Series, Vol. 34, Princeton University Press, Princeton, NJ, 2010.

\bibitem[GLO11]{ger-leb-obl:quantum} A. Gerasimov, D. Lebedev, and S. Oblezin, Quantum Toda chains intertwined, St. Petersburg Math. J. {\bf  22} (2011), 411--435.

\bibitem[GLO12]{ger-leb-obl:new}    \bysame, 
New integral representations of Whittaker functions for classical Lie groups,  Russian Math. Surveys {\bf 67} (2012), 1--92. 


\bibitem[GW86]{goo-wal:classical} R. Goodman and N.R. Wallach, Classical and quantum mechanical systems 
of Toda-Lattice type III. Joint eigenfunctions of the quantized systems, Commun. Math. Phys. {\bf 105} (1986), 473--509.

\bibitem[G01]{gru:bispectral} F.A. Gr\"unbaum,
The bispectral problem: an overview, in: Special Functions 2000: Current Perspective and Future Directions,
J. Bustoz, M.E.H. Ismail, and S.K. Suslov (eds.),
NATO Sci. Ser. II Math. Phys. Chem., Vol.  {30}, Kluwer Acad. Publ., Dordrecht, 2001,  129--140.

\bibitem[H09]{hal:multivariable} M. Halln\"as, Multivariable Bessel polynomials related to the hyperbolic Sutherland model with external Morse potential, 
Int. Math. Res. Not. IMRN {\bf 2009}(9), 1573--1611.

\bibitem[HL10]{hal-lan:unified} M. Halln\"as and E. Langmann, A unified construction of generalized classical polynomials associated with operators of Calogero-Sutherland type, Constr. Approx. {\bf 31} (2010), 309--342.

\bibitem[HR12]{hal-rui:kernel} M. Halln\"as and S.N.M. Ruijsenaars, Kernel functions and B\"acklund transformations for relativistic Calogero-Moser and Toda systems, J. Math. Phys. {\bf 53} (2012), 123512.

\bibitem[HR15]{hal-rui:recursive} \bysame, A recursive construction of joint eigenfunctions for the hyperbolic nonrelativistic Calogero-Moser Hamiltonians, Int. Math. Res. Not. IMRN {\bf 2015}(20), 10278--10313.

\bibitem[H82]{has:whittaker} M. Hashizume,
Whittaker functions on semisimple Lie groups,
Hiroshima Math. J. {\bf 12} (1982),  259--293. 

\bibitem[HO87]{hec-opd:root} G.J. Heckman and E.M. Opdam, Root systems and hypergeometric functions. I, Compos. Math. {\bf 64} (1987), 329--352.

\bibitem[HS94]{hec-sch:harmonic} G. Heckman and H. Schlichtkrull,
Harmonic Analysis and Special Functions on Symmetric Spaces,
Perspectives in Mathematics, Vol. 16, Academic Press, Inc., San Diego, CA, 1994.

\bibitem[H00]{hel:groups} S. Helgason,
Groups and Geometric Analysis: 
Integral Geometry, Invariant Differential Operators, and Spherical Functions, Mathematical Surveys and Monographs, Vol. 83, American Mathematical Society, Providence, RI, 2000.


\bibitem[I89]{ino:finite} V.I. Inozemtsev, The finite Toda lattices, Comm. Math. Phys. {\bf 121} (1989), 629--638.

\bibitem[IM86]{ino-mes:discrete} V.I. Inozemtsev and D.V. Meshcheryakov,
The discrete spectrum states of finite-dimensional quantum systems connected with Lie algebras,
Phys. Scripta {\bf 33} (1986), 99--104. 

\bibitem[I13]{ish:whittaker} T. Ishii, Whittaker functions on orthogonal groups of odd degree, J. Lie Theory {\bf 23} (2013), 85--112.

\bibitem[J06]{jeu:paley-wiener} M. de Jeu, Paley-Wiener theorems for the Dunkl transform, 
Trans. Amer. Math. Soc. {\bf 358} (2006), 4225--4250.

\bibitem[KL01]{kha-leb:integral} S. Kharchev and D. Lebedev,
Integral representations for the eigenfunctions of quantum open and periodic Toda chains from the QISM formalism,
J. Phys. A {\bf 34} (2001), 2247--2258. 

\bibitem[K79]{kos:quantization} B. Kostant, Quantization and representation theory, in:
Representation Theory of Lie groups (G.L. Luke, ed.), London Mathematical Society Lecture Note Series, Vol. 34, Cambridge University Press, Cambridge-New York, 1979,  287--316.

\bibitem[K13]{koz:aspects}  K.K. Kozlowski, Aspects of the inverse problem for the Toda chain, J. Math. Phys. {\bf 54} (2013), 121902.

\bibitem[K97]{kuz:separation} V.B. Kuznetsov, Separation of variables for the $\mathcal{D}_n$-type periodic Toda lattice, J. Phys. A {\bf 30} (1997), 2127--2138.

\bibitem[KJC95]{kuz-jor-chr:new} V.B. Kuznetsov, M.F. J\o{}rgensen, and P.L. Christiansen,
New boundary conditions for integrable lattices,
J. Phys. A {\bf 28} (1995), 4639--4654. 

\bibitem[OP83]{ols-per:quantum} M.A. Olshanetsky and A.M. Perelomov,
Quantum integrable systems related to Lie algebras, 
Phys. Rep. {\bf 94} (1983), 313--404. 

\bibitem[O88]{opd:root} E.M. Opdam, Root systems and hypergeometric functions. IV, Compos. Math. {\bf 67} (1988), 191--209.

\bibitem[O95]{opd:harmonic} \bysame, Harmonic analysis for certain representations of graded Hecke algebras, Acta Math. {\bf 175} (1995), 75--121.

\bibitem[O00]{opd:lecture} \bysame,
Lecture Notes on Dunkl Operators for Real and Complex Reflection Groups,
MSJ Memoirs, Vol. 8, Mathematical Society of Japan, Tokyo, 2000. 

\bibitem[O07]{osh:completely} T.  Oshima, Completely integrable systems associated with classical root systems, SIGMA Symmetry Integrability Geom. Methods Appl. {\bf 3} (2007), Paper 061. 

\bibitem[OS10]{osh-shi:heckman-opdam} T. Oshima and N. Shimeno,
Heckman-Opdam hypergeometric functions and their specializations, in: New Viewpoints of Representation Theory and Noncommutative Harmonic Analysis, M. Itoh and H. Ochiai (eds.),
RIMS K\^{o}ky\^{u}roku Bessatsu, Vol. B20, Res. Inst. Math. Sci. (RIMS), Kyoto, 2010,  129--162.

\bibitem[P12]{pus:hyperbolic} B.G. Pusztai,
The hyperbolic $BC_n$ Sutherland and the rational $BC_n$ Ruijsenaars-Schneider-van Diejen models: Lax matrices and duality,
Nuclear Phys. B {\bf 856} (2012), 528--551. 

\bibitem[P13]{pus:scattering} \bysame, 
Scattering theory of the hyperbolic $BC_n$ Sutherland and the rational $BC_n$ Ruijsenaars-Schneider-van Diejen models,
Nuclear Phys. B {\bf 874} (2013), 647--662. 

\bibitem[R12]{rie:mirror}  K. Rietsch, A mirror symmetric solution to the quantum Toda lattice, Comm. Math. Phys. {\bf 309} (2012),  23--49. 

\bibitem[RKV13]{ros-koo-voi:limit} M. R\"osler, T. Koornwinder, and M.  Voit,
Limit transition between hypergeometric functions of type BC and type A, 
Compos. Math. {\bf 149} (2013), 1381--1400. 

\bibitem[RV04]{ros-voi:positivity} M. R\"osler and M. Voit, Positivity of Dunkl's intertwining operator via the trigonometric setting,
 Int. Math. Res. Not. IMRN {\bf 2004}(63), 3379--3389.
 
\bibitem[RV16]{ros-voi:integral} \bysame,  Integral representation and uniform limits for some Heckman-Opdam hypergeometric functions of type BC, Trans. Amer. Math. Soc. {\bf 368} (2016), 6005--6032.

\bibitem[R87]{rui:complete} S.N.M. Ruijsenaars, Complete integrability of relativistic Calogero-Moser systems and elliptic function identities, Comm. Math. Phys. {\bf 110} (1987), 191--213. 

\bibitem[R88]{rui:action-angle} \bysame,
Action-angle maps and scattering theory for some finite-dimensional integrable systems. I. The pure soliton case. 
Comm. Math. Phys. {\bf 115} (1988), 127--165. 

\bibitem[R90a]{rui:relativistic} \bysame, Relativistic Toda systems, Comm. Math. Phys. {\bf 133} (1990), 217--247. 

\bibitem[R90b]{rui:finite-dimensional} \bysame, Finite-dimensional soliton systems, in: {Integrable and
Superintegrable Systems}, B.A. Kupershmidt (ed.), World Scientific Publishing Co., Inc., Teaneck, NJ, 1990, 165--206.

\bibitem[S99]{saw:eigenfunctions} P. Sawyer, The eigenfunctions of a Schr\"odinger operator associated to the root system $A_{n-1}$,
Quart. J. Math. Oxford Ser. (2) {\bf 50} (1999), 71--86.

\bibitem[S08]{shi:limit} N. Shimeno, A limit transition from Heckman-Opdam hypergeometric functions to the Whittaker functions associated with root systems, arXiv.0812.3773.


\bibitem[S88]{skl:boundary} E.K. Sklyanin,
Boundary conditions for integrable quantum systems,
J. Phys. A {\bf 21} (1988), 2375--2389.

\bibitem[S13]{skl:bispectrality} \bysame, Bispectrality for the quantum open Toda chain,
J. Phys. A {\bf 46} (2013), 382001.

\bibitem[T98]{tsi:dynamical} A.V.  Tsiganov, Dynamical boundary conditions for integrable lattices, J. Phys. A {\bf 31} (1998), 8049--8061.

\bibitem[W84]{woj:integrability} S.  Wojciechowski,
On the integrability of the Calogero-Moser system in an external quartic potential and other many-body systems,
Phys. Lett. A {\bf 102 } (1984), 85--88. 


\end{thebibliography}

\end{document}